\documentclass[UKenglish,cleveref, autoref, thm-restate, algorithmic, algorithm]{lipics-v2021}
\usepackage[ruled, noend, algo2e]{algorithm2e}
\usepackage[noend]{algorithmic} 

\usepackage{amsmath, amsfonts, amssymb}
\usepackage{todonotes}
\usepackage{booktabs}
\usepackage{numprint}
\usepackage{changepage}
\usepackage{tcolorbox}
\usepackage{comment}
\usepackage{enumitem}
\usepackage{enumerate}
\usepackage{enumitem}

\usepackage[table]{xcolor}

\usepackage[thinlines]{easytable}
\usepackage{makecell}
\usepackage{arydshln}
\nolinenumbers

\usepackage{hyperref}

\usepackage[T1]{fontenc}
\usepackage{tikz}
\usetikzlibrary{shapes.geometric, arrows.meta, positioning, calc, backgrounds}

\newcommand{\set}[1]{\{ #1 \}}

\newcommand{\R}{\mathbb{R}}
\newcommand{\N}{\mathbb{N}}

\newcommand{\E}{\mathbb{E}}

\newtheorem{invariant}{Invariant}

\definecolor{colHard}{RGB}{250, 235, 235} 
\definecolor{colTwo}{RGB}{235, 245, 255}  
\definecolor{colOne}{RGB}{235, 255, 235}  
\definecolor{colBorder}{RGB}{0, 0, 0}     

\tikzset{
    problem/.style={
        draw=black, 
        fill=white, 
        rectangle, 
        rounded corners=2pt, 
        align=center, 
        minimum height=1.8em, 
        inner sep=3pt,
        font=\scriptsize\sffamily, 
        line width=0.5pt
    },
    classbox/.style={
        draw=colBorder, 
        line width=0.8pt,
        rectangle,
        rounded corners=0pt,
        inner sep=0pt,       
        minimum width=10.5cm, 
        align=center
    },
    classheader/.style={
        anchor=north west, 
        font=\scshape, 
        text=black,
        xshift=3pt,
        yshift=-3pt
    },
    reduction/.style={
        -{Stealth[length=2mm, width=1.5mm]}, 
        thick, 
        draw=black!70,
        smooth
    }
}

\ccsdesc[500]{Theory of computation~Computational geometry}

\keywords{Planar geometry, data structures, algebraic decision trees, algorithms with advice.}

\title{When is presorting enough? Breaking the $O(n \log n)$ barrier for geometric data structures} 
\titlerunning{When is presorting enough? Breaking the $O(n \log n)$ barrier for geometric data structures}

\title{The Presort Hierarchy for Geometric Problems} 
\titlerunning{The Presort Hierarchy for Geometric Problems}

\author{Ivor van der Hoog}{IT University of Copenhagen, Denmark}{ivva@itu.dk}{https://orcid.org/0009-0006-2624-0231}{}

\author{Eva Rotenberg}{IT University of Copenhagen, Denmark}{erot@itu.dk}{https://orcid.org/0000-0001-5853-7909}{}

\author{Jack Spalding-Jamieson}{Independent}{jacksj@uwaterloo.ca}{https://orcid.org/0000-0002-1209-4345}{}

\author{Lasse Wulf}{IT University of Copenhagen, Denmark}{lasw@itu.dk}{https://orcid.org/0000-0001-7139-4092}{}

\newcommand{\mysubpara}[1]{%
  \par\vspace{0.15\baselineskip}%
  \noindent\textsf{\textbf{#1}}\hspace{0.5em}%
}
\hideLIPIcs
\authorrunning{Ivor van der Hoog, Eva Rotenberg, Jack Spalding-Jamieson, Lasse Wulf} 

\acknowledgements{We are thankful to Da Wei Zheng for helpful discussions on extending and simulating Word RAM operations. We are thankful to Jeff Erickson for pointing us to the lower bound of Seidel for the 3D Convex Hull problem.}

\funding{This work was supported by the the VILLUM Foundation grant (VIL37507)
``Efficient Recomputations for Changeful Problems''.}

\begin{document}

\maketitle

\begin{abstract}
Many fundamental problems in computational geometry admit no algorithm running in $o(n \log n)$ time for $n$ planar input points, via classical reductions from sorting. Prominent examples include the computation of convex hulls, quadtrees, onion layer decompositions, Euclidean minimum spanning trees, KD-trees, Voronoi diagrams, and decremental closest-pair.

A classical result shows that, given $n$ points sorted along a single direction, the convex hull can be constructed in linear time. Subsequent works established that for all of the other above problems, this information does not suffice.
In 1989, Aggarwal, Guibas, Saxe, and Shor asked: Under which conditions can a Voronoi diagram be computed in $o(n \log n)$ time? Since then, the question of whether sorting along \emph{two} directions enables a $o(n \log n)$-time algorithm for such problems has remained open and has been repeatedly mentioned in the literature.

In this paper, we introduce the Presort Hierarchy: A problem is \textbf{\texttt{1-Presortable}} if, given a sorting along one axis, it permits a (possibly randomised) $o(n \log n)$-time algorithm. It is \textbf{\texttt{2-Presortable}} if sortings along \emph{both} axes suffice. It is \textbf{\texttt{Presort-Hard}} otherwise.
Our main result is that quadtrees, and by extension Delaunay triangulations, Voronoi diagrams, and Euclidean minimum spanning trees, are \textbf{\texttt{2-Presortable}}: we present an algorithm with expected running time $O(n \sqrt{\log n})$.
This addresses the longstanding open problem posed by Aggarwal, Guibas, Saxe, and Shor (albeit randomised). 
We complement this result by showing that some of the other above geometric problems are also \textbf{\texttt{2-Presortable}}  or \textbf{\texttt{Presort-Hard}}.
\end{abstract}

\setcounter{page}{0}

\newpage
\section{Introduction}
Many problems in computational geometry have tight algorithms whose worst-case running time is known to be optimal at $\Theta(n \log n)$. 
For planar input, prominent examples include (but are not limited to) Pareto front, convex hulls, quadtrees, Delaunay triangulations, Voronoi diagrams, well-separated pair decompositions, Euclidean minimum spanning trees, planar triangulations, onion layer decompositions, KD-trees, and decremental closest-pair structures.

The above problems have $\Omega(n \log n)$ \emph{comparison-based} lower bounds: if, for inputs of size $n$, a problem admits $N$ distinct outputs, then any comparison-based model of computation (such as decision trees or the Real RAM) has an $\Omega(\log N)$ lower bound~\cite{Ben-Or83}. 
A natural question is under which conditions these planar data structures can be constructed in time faster than $O(n \log n)$. 
Andrew~\cite{andrew1979another} computes a planar convex hull in linear time when the points are given sorted by $x$-coordinate; the same idea yields a linear-time algorithm for the Pareto front. 
In contrast, Seidel~\cite{SEIDEL1985319} showed that in $\mathbb{R}^3$, an $\Omega(n \log n)$ lower bound persists 
even when the input is presorted along any arbitrarily large constant number of directions.

Aggarwal, Guibas, Saxe, and Shor~\cite{AggarwalGuibasSaxeShor1989} presented in 1989 a linear-time algorithm for constructing the Voronoi diagram of a convex polygon. Since the convexity of the input point set $P$ implies a strong form of presorting --- a known radial order along the boundary --- they asked which other kinds of presorting information suffice.
Chew and Fortune~\cite{Chew-Fortune-1997} subsequently showed that orthogonal presorting can lead to faster algorithms, but their argument works only for nonstandard Voronoi diagrams. In particular, they give an $O(n \log \log n)$-time algorithm, assuming that the input points are presorted along both the $x$- and $y$-coordinates, for constructing Voronoi diagrams under convex distance functions where balls are triangle-shaped. They explicitly ask whether the Euclidean Voronoi diagram can be constructed in $o(n \log n)$ time given sortings along two orthogonal directions.
Djidjev and Lingas~\cite{DJIDJEV1995VoronoiSorted} showed a comparison-based $\Omega(n \log n)$ lower bound for constructing the Euclidean Voronoi diagram even when the input points are sorted by $x$-coordinate. They also show the existence of a  clairvoyant algebraic decision tree (we formally define models of computation later) that, given the input points presorted along two directions, determines the Voronoi diagram with linear depth. This demonstrates that the standard comparison-based arguments do not rule out faster algorithms, and it motivates the following longstanding open question:

\vspace{-.5em}

\begin{quote}
Does there exist an $o(n \log n)$-time algorithm to construct the Euclidean Voronoi diagram of a planar point set that is sorted along both $x$ and $y$?
\end{quote}

\vspace{-.5em}

\mysubpara{Equivalence.}
By classical planar duality, a planar Delaunay triangulation can be transformed into a Voronoi diagram, and vice versa, in deterministic linear time.
Krznari\'{c} and Levcopoulos~\cite{Krznaric1998Computing} showed that a Delaunay triangulation can be transformed into a quadtree in linear time.
Buchin and Mulzer~\cite{BuchinMulzer2011Delaunay} gave a randomised linear-time algorithm for constructing a Delaunay triangulation from a quadtree.
Subsequently, L\"{o}ffler and Mulzer~\cite{Loffler2012Triangulating} established deterministic linear-time algorithms for both transformations.
Once a Delaunay triangulation is available, the Euclidean minimum spanning tree (EMST) can be computed in linear time, since the EMST is a subgraph of the Delaunay triangulation.
L\"{o}ffler and Mulzer~\cite{Loffler2012Triangulating} also show the other direction, as they can construct a Delaunay triangulation from an EMST in linear time. Finally, they also show that quadtrees and well-separated pair decompositions (WSPDs) are equivalent in linear time.
Taken together, these results imply that quadtrees, Delaunay triangulations, Voronoi diagrams, WSPDs, and EMSTs form a tightly connected family of geometric data structures: a sub-$O(n \log n)$-time algorithm for constructing any one of them immediately yields such an algorithm for all the others.
We refer to these structures as \emph{proximity structures}.
Under this equivalence, the open question can be phrased as follows:

\begin{quote}
    Does there exist an $o(n \log n)$-time algorithm to construct any proximity structure of a planar point set that is sorted along both the $x$- and $y$-coordinates?
\end{quote}

We discuss further related work in Appendix~\ref{app:closely}.

\mysubpara{Contribution and results.}
We answer the above question in a positive manner when allowing for randomisation.
In full generality, we argue for a \emph{hierarchy} of presortability.
A problem is called \texttt{1-Presortable} if it admits a (randomised) $o(n \log n)$-time algorithm when the input is sorted by $x$. 
A problem is \texttt{2-Presortable} if it is not \texttt{1-Presortable} and admits a (possibly randomised) $o(n \log n)$-time algorithm when the input is given sorted along both the $x$- and $y$-directions, together with the permutation relating these orderings.
Finally, a problem is \texttt{Presort-Hard} if it is not \texttt{2-Presortable}.
We note that the lower bound of Seidel~\cite{SEIDEL1985319} applies even when the input is sorted along any constant number of directions.
The same is true for the lower bounds proved in this paper; however, for ease of exposition, we adopt the above, less general definition of \texttt{Presort-Hard}.

Our main contribution is to show that constructing a quadtree is \texttt{2-Presortable} on the Real RAM, if randomisation is allowed.
Specifically, we present an algorithm with expected running time $O(n \sqrt{\log n})$ (Theorem~\ref{thm:quadtree}).
Combined with the known linear-time reductions between proximity structures, this result resolves the open problems posed by Aggarwal, Guibas, Saxe, and Shor~\cite{AggarwalGuibasSaxeShor1989} and by Djidjev and Lingas~\cite{DJIDJEV1995VoronoiSorted}, when randomisation is permitted.

We complement this positive result by placing several related geometric problems into this hierarchy (see Table~\ref{tab:results}).
We show that our result implies that finding the maximum empty circle is \texttt{2-Presortable}. 
We also show that KD-trees are \texttt{2-Presortable}. Finally, we show that orthogonal segment intersection is in this class.
We prove that computing the onion layer decomposition is \texttt{Presort-Hard}, which in turn implies that decremental convex hulls are \texttt{Presort-Hard}. As a consequence, onion layer decompositions do not belong to the family of proximity structures. Finally, we show that ordered $k$-closest pairs (and therefore decremental closest pairs) are \texttt{Presort-Hard}.

{
\begin{table}[h]
\centering
\small
\renewcommand{\arraystretch}{0.95}
\setlength{\dashlinedash}{0.5pt}
\setlength{\dashlinegap}{1.5pt}
\setlength{\arrayrulewidth}{0.4pt}
\begin{tabular}{|p{0.33\linewidth}|>{\centering\arraybackslash}p{0.40\linewidth}|>{\centering\arraybackslash}p{0.17\linewidth}|}
\hline
\textbf{Problem Name} & \textbf{Proof technique} & \textbf{Reference}\\
\hline
\rowcolor{green!20}
\multicolumn{3}{|c|}{\textbf{\texttt{Presort-Hard}}}\\
\hdashline
3D Convex Hull & Direct & \cite{SEIDEL1985319}\\
\hdashline
Ordered $k$-Closest-Pairs & Reduction~from Sorting & \cref{thm:k-closest-pair}\\
\hdashline
Decremental Closest-Pair & Reduction~from Ordered $k$-Closest-Pairs & \cref{cor:decremental-closest-pair}\\
\hdashline
Onion Layer Decomposition & Direct & \cref{thm:onion_layer}\\
\hdashline
Decremental Convex Hull & Reduction~from Onion Layer Decomp. & \cref{cor:decremental-convex-hull}\\
\hline
\rowcolor{green!20}
\multicolumn{3}{|c|}{\textbf{\texttt{$2$-Presortable}}}\\
\hdashline
Quadtree & Direct & \cref{thm:quadtree}\\
\hdashline
NN-graph & Reduction~to Quadtrees & \cref{cor:proximity}, \cite{Loffler2012Triangulating}\\
\hdashline
Delaunay Triangulation & Reduction~to NN-graphs & \cref{cor:proximity}, \cite{Loffler2012Triangulating}\\
\hdashline
Euclidean MST & Reduction~to Delaunay Triangulations &  \cref{cor:proximity}, \cite{Loffler2012Triangulating}\\
\hdashline
Voronoi Diagram & Reduction~to Delaunay Triangulations & \cref{cor:proximity}, \cite{Loffler2012Triangulating}\\
\hdashline
Maximum Empty Circle & Reduction to Delaunay Triangulations & \cref{thm:maximum}   \\ 
\hdashline
KD-tree & Reduction to Quadtrees & \cref{thm:kd-tree}\\
\hdashline
Orthogonal Segment Intersection & Direct & \cref{thm:isect} \\
\hline
\rowcolor{green!20}
\multicolumn{3}{|c|}{\textbf{\texttt{$1$-Presortable}}}\\
\hdashline
2D Convex Hull & Direct & \cite{andrew1979another} or \cite{graham72} \\
\hdashline
2D Pareto Front & Direct & adapting \cite{andrew1979another}, \cite{graham72}  \\
\hdashline
Diameter of a Point set & Reduction~to 2D Convex Hull & \cite{andrew1979another} or \cite{graham72} \\
\hdashline
A triangulation & Direct & \cref{thm:a-triangulation} \\
\hline
\end{tabular}
\caption{The known elements in our hierarchy. 
All the results with labelled theorems or corollaries are novel.
We also state whether a result is proven directly, or where it is derived from.
\vspace{-0.7cm}}
\label{tab:results}
\end{table}
}

\section{Preliminaries}
\label{sec:prelim}

In this paper, the algorithmic input consists of both real-valued and integer data.
In full generality, the real-valued input is a sequence
$A = (a_1, \ldots, a_r)$ of length $r$,
and the integer input is a sequence
$B = (b_1, \ldots, b_m)$ of length $m$.

\mysubpara{Presorting.}
In this paper, the integer input encodes a \emph{promise} about the real-valued data.
One such promise is what we call \emph{presortedness}.
We say that the input $(A,B)$ is a $1$-\emph{presorting} of a point set $P$ if $A$ stores $P$ sorted along a direction in $\mathbb{R}^2$. 
We say that the input $(A, B)$ is a $2$-\emph{presorting} of a point set $P$ if $A$ stores $P$ twice (sorted along orthogonal directions) and $B$ specifies the permutation between the two sortings. 
Formally, we say that $r = 2 \cdot n$ and that the real-valued sequence $A$ can be partitioned as $
A = A_x \cup A_y$ where $A_x$ sorts $P$ by $x$ and $A_y$ sorts $P$ by $y$. 
For notational convenience,  we index the points of $P$ so that
$A_x[i] = p_i$ for all $i \in [n]$.
We assume that the integer input $B$ then specifies the permutation from $A_x$ to $A_y$. 
Formally, it is an integer sequence $\pi$ such that $p_i = A_y[\pi(i)]$ for all $i \in [n]$.
We study for which geometric problems it is useful to have access to a presorting.

\mysubpara{The Real RAM.}
We provide both upper and lower bounds for geometric problems whose input contains both real-valued and integer data.
It is therefore necessary to define our computational models with some care.
Our primary model for upper bounds is the standard model for real-valued geometric computation, namely the \emph{Real RAM}.
We adopt the textbook definition of the Real RAM due to Preparata and Shamos~\cite{preparata2012computational}.
This classical model allows both real-valued and integer operations: it is a Random Access Machine (RAM) in which each memory cell may store either a real number or an integer.
The input $x \in \mathbb{R}^{r} \times \mathbb{N}^m$ occupies the first $r + m$ memory cells, in order.
A program operating on input $x$ is a finite sequence of instructions, each of which is of one of the following three types:\vspace{-0.3cm}
\begin{itemize}
    \item \emph{Primitives} manipulate storage cells, after which execution proceeds to the next instruction.
    \item \emph{Control-flow} specifies an integer $i$ and reads a memory cell containing a Boolean; if \texttt{True}, execution jumps to the $i$th instruction, and otherwise proceeds to the next instruction.
    \item \emph{Halts} terminate the program, and the memory at that time constitutes the output.
\end{itemize}
\noindent
A \emph{Primitive} takes two memory cells as input and writes its result to a new memory cell. The model supports the following constant-time Primitives:
\begin{itemize}\vspace{-0.3cm}
    \item Basic arithmetic operations (addition $+$, subtraction $-$, multiplication $\times$, and division $/$).
    \item Comparison operations (equality $=$ and order $<$).
    \item Indirect addressing, but only via memory cells that are designated to contain integers.
\end{itemize}
\noindent
These Primitives are subject to the following restrictions.
Real values cannot be written to integer memory cells.
Any arithmetic operation involving at least one real-valued cell must write its result to a real-valued cell; in particular, such results cannot be used for memory addressing.
Comparison operations output a Boolean, which is stored as an integer and may be used by Control-flow instructions for branching.
Integer division $x,y \mapsto \left\lfloor \frac{x}{y} \right\rfloor$ is \emph{not} a Primitive, even when both inputs are integers.
Finally, we allow randomisation, which can be modelled as a Control-flow branch that is taken with probability $\frac12$.

\mysubpara{Algebraic Decision Trees.}
To obtain Real RAM lower bounds, we typically apply a Ben-Or~\cite{Ben-Or83} argument: counting the leaves in a real-valued \emph{Algebraic Decision Tree} (ADT). 

For a \emph{real-valued} ADT with input-size $N$, the input is a vector $X = (x_1, \ldots, x_{N})$ of $N$ real numbers.
It is a rooted tree in which each internal node is labelled with a multivariate polynomial
$p(x_1, \ldots, x_{N}) \in \mathbb{R}[x_1, \ldots, x_{N}]$.
Each internal node has three outgoing edges labelled
$p(x) < 0$, $p(x) = 0$, and $p(x) > 0$.
Each leaf of the ADT contains a \emph{combinatorial output}, that is, a list of integers, where each integer refers to an index of an input variable.
For example, when sorting the $N$ numbers, each leaf contains a permutation of $[N]$, where a permutation such as $(1,3,7,\ldots)$ indicates that $x_1$ is the smallest input, followed by $x_3$, then $x_7$, and so on.
An ADT is \emph{complete} if every fixed input $X$ traverses a path to a leaf. 
It is \emph{correct} if it is complete and every input leads to the corresponding desired output. 

If the input to a Real RAM program is purely real-valued, then real-valued ADTs are strictly more powerful than the Real RAM.
This follows from a classical argument where any real-valued Real RAM program constructs an ADT whose depth is at most its worst-case running time~\cite{Ben-Or83}.
Therefore, any lower bound for real-valued ADTs also applies to the Real RAM when the input is purely real-valued.
Since each internal node has three outgoing edges, an ADT with $k$ leaves must have depth at least $\Omega(\log_3 k)$ on some input.
Thus, if over all real-valued inputs $X \in \mathbb{R}^{N}$ there are $k$ distinct combinatorial outputs, then both real-valued ADTs and the Real RAM admit a lower bound of $\Omega(\log_3 k)$. This also holds in expectation, since the median depth of the tree is $\Omega(\log_3 k)$.

\mysubpara{Integer inputs and ADTs.}
In many cases, however, the algorithmic input also contains integers (see, for example, the lower bounds by Seidel~\cite{SEIDEL1985319}).
If the input contains integers, the computational strength of a real-valued ADT and the Real RAM are incomparable: as the Real RAM can use the integer input for indirect addressing and outputting a solution.
We give a concrete example: let the real input $A \subset \mathbb{R}^n$ specify a real-valued point set $P$ and the integer input $\pi \subset \mathbb{N}^n$ specify the sorting of $P$ along a space-filling curve (for each index $i$,  the element $A[\pi(i)]$ is the $i$'th element along the sorting order). 
If the goal is to sort $P$ along this space-filling curve, then a Real RAM program can simply output $\pi$. 
Yet, any program that can only receive real-valued input is required to sort $A$.

\mysubpara{Clairvoyant ADTs.}
As a consequence, lower bounds for real-valued ADTs are not lower bounds for the Real RAM when presortings or other integer-valued promises are provided, and we have to adopt a different strategy.
Let the input be a real-valued vector $A$ and some integer vector $B$. For any fixed $B$, we define a \emph{clairvoyant} ADT $\mathcal{A}_B$ in a similar way to~\cite{Afshani2017Instance, Hoog2025Convex}: It is a real-valued ADT that knows $B$ as a constant and receives $A$ as input.
For any Real RAM program with input $A \cup B$, a clairvoyant ADT $\mathcal{A}_B$ is strictly more powerful. In particular,  consider any lower bound $L$ on the maximum depth traversed by a clairvoyant ADT $\mathcal{A}_B$ over all real-valued inputs $A$ such that $A \cup B$ is a valid input.
Then $L$ is also a worst-case lower bound for any Real RAM program.
Clairvoyant ADTs are thus a powerful tool for lower bounds. For upper bounds they are less applicable because they are computationally quite strong: any Real RAM or ADT cannot simulate a clairvoyant ADT $\mathcal{A}_B$ since it would have to identify the input $B$ and `find' the corresponding program or tree.

\mysubpara{Word RAM.}
Finally, we mention the Word RAM model as given by Fredman~\cite{Fredman1993Surpassing}.
The reason we mention this model is that we will be simulating its instructions on a Real RAM (under some assumptions about what those operations operate on). Let the input size be $N$. 
In the Word RAM model, we have a machine with an unbounded sequence of memory cells, each storing an integer of $O(\log N)$ bits.
A program is again a finite sequence of instructions, of \emph{Primitives}, \emph{Control-flow} and \emph{Halts}.  A \emph{Primitive} takes $k$ memory cells as input and writes its result to a memory cell.
The model supports the following constant-time Primitives:
\begin{itemize}[nolistsep]
    \item Integer arithmetic (addition $+$, subtraction $-$, multiplication $\times$, and integer division).
    \item Comparison operations (equality $=$ and order $<$).
    \item Direct and indirect memory addressing.
\end{itemize}
\noindent
Each Primitive operating on $k$ memory cells takes $k$ time. In our use of this model,
we will only store integers in $[2n]$, and $k$ will always be constant.
As we will briefly discuss in \cref{sec:auxiliary},
many other word-level operations such as finding the most significant bit can, after linear-time preprocessing, be simulated in constant time~\cite{Fredman1993Surpassing}.


\section{Quadtrees via Linear-depth Clairvoyant Algebraic Decision Trees}

Buchin and Mulzer~\cite{BuchinMulzer2011Delaunay} mention the existence of
an algebraic decision tree of linear depth to construct a quadtree of a presorted point set $(A_x, A_y, \pi)$,
although they do not provide details.
We note that via the linear-time equivalence between a Voronoi diagram and a quadtree, the result by Djidjev and Lingas~\cite{DJIDJEV1995VoronoiSorted} indirectly implies a linear-depth clairvoyant ADT. 
Via private communication, we confirmed that Buchin and Mulzer \cite{BuchinMulzer2011Delaunay} envisioned an alternative and direct construction of a linear-depth clairvoyant ADT.
Details of their construction will be helpful for our algorithm in Section~\ref{sec:sub-n-log-n-algo}. 
Hence, in this section, we present an explicit construction:

\begin{theorem}[{Buchin and Mulzer \cite{BuchinMulzer2011Delaunay}}]
\label{thm:buchin}
    Let $\pi \colon [n] \rightarrow [n]$ be a permutation. Denote by $\mathbb{I}_\pi$ all presorted point sets $(A_x \cup A_y) \times \pi$ where $\forall i \in [n]$, $A_x[i] = A_y[\pi(i)]$. 
    There exists a clairvoyant decision tree $\mathcal{T}_\pi$ of linear depth that for all $(A_x, A_y) \in \mathbb{I}_\pi$ finds the corresponding quadtree.
\end{theorem}

\mysubpara{Open and closed squares.}
A quadtree is, on a high level, a partitioning of a bounding square $B$ into smaller squares such that each square contains at most one point of the input $P$. 
To formally define such a partition, we need to distinguish between open and closed squares.
A closed square is the area $[a, a+c] \times [b, b+c]$ for positive integers $a$, $b$ and $c$. 
An open square is the area $[a, a+c) \times [b, b+c)$ and we say that the squares $[a, a+c) \times [b, b+c]$ and $[a, a+c] \times [b, b+c)$ are \emph{half-open}.
In the definition of a (compressed) quadtree, we must carefully partition the bounding square $B$ into closed, open, and half-open squares such that each point of $P$ lies in exactly one square of the resulting partition.

\mysubpara{Quadtrees.}
Let $P$ be a planar point set contained in some (open, half-open, or closed) axis-aligned square $B$.
The \emph{split} operation $\texttt{split}(B)$ is the unique partitioning of $B$ into four equal-area closed, half-open and open squares such that each point $q \in B$ is contained in a unique part.
A \emph{quadtree} is a tree obtained by recursively applying the split operation on squares $C$ in $B$ where $P \cap C$ contains at least two points.
In the resulting tree $T$, we do not distinguish between nodes $v$ of $T$ and their corresponding square in the plane.  
The depth of the resulting tree is $\Theta(n \log \Psi)$, where $\Psi$ is the ratio between the distance realised by the farthest pair of points and the closest pair of points. Since $P$ may be an arbitrary size-$n$ planar point set, this quantity is not necessarily bounded.

\mysubpara{Compressed quadtrees.}
To avoid having unbounded-depth trees, we typically define and compute a \emph{compressed quadtree}~\cite{Loffler2012Triangulating}. Morally, a compressed quadtree considers maximal descending paths $v_0 \rightarrow v_1 \rightarrow \ldots \rightarrow v_\ell$
in the quadtree such that $v_1 \cap P = v_\ell \cap P$ and, if this path has length 2 or greater, makes $v_0$ the parent of $v_\ell$. 
Formally, we are unable to compute this quadtree cell $v_k$ on a Word RAM as it effectively corresponds to computing the maximum integer $k$ such that there exist integers $(i, j)$ such that (a subset of) $P$ is contained in the square $[i \cdot 2^{-k}, (i+1) \cdot 2^{k}) \times [j \cdot 2^{-k}, (j+1) \cdot 2^{-k})$. This in turn requires integer arithmetic operations that are not available on the Real RAM. 
Instead, compressed quadtrees are recursively defined:

\begin{definition}[Figure~\ref{fig:quadtree-split}]
Let $B$ be an (open, half-open, or closed) axis-aligned bounding box containing a planar point set $P$.
Let $(B_1, B_2, B_3, B_4) = \texttt{split}(B)$. 
The \emph{compressed} quadtree $T(B, P)$ has as its root $B$ which is split according to one of two types:
\begin{enumerate}[label=\roman*.]
    \item \label{type:one} If $P$ is entirely contained in one of the four squares $B_i$, then let $B' \subset B_i$ be a minimum-area closed square containing $P$. The root $B$ has a single child which is $T(B', P)$.
    \item \label{type:two} Otherwise, $B$ has four children which are $T(B_i, P \cap B_i)$ for $i \in [4]$.
\end{enumerate}
\end{definition}
\begin{figure}
    \centering
    \includegraphics[]{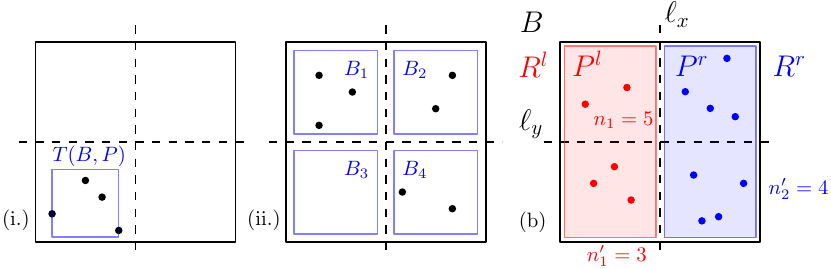}
    \caption{Depiction of type i.\ and type ii.\ quadtree splits. (b) Our variable names. \vspace{-0.3cm}}
    \label{fig:quadtree-split}
\end{figure}

\subsection{A Linear-depth Clairvoyant Algebraic Decision Tree}
\label{sub:adt}

Constructing a compressed quadtree reduces to sorting, and therefore there cannot exist a Real RAM algorithm whose (expected) running time is $o(n \log n)$ without additional information.
Yet (Theorem~\ref{thm:buchin}),  for a fixed permutation $\pi$ there exists a clairvoyant ADT $\mathcal{T}_\pi$ of linear depth such that for all $2$-presortings $(A_x \cup A_y) \times \pi $, the ADT finds the quadtree of $P$.
This is \emph{not} sufficient for an algorithm which constructs the quadtree for arbitrary pointsets with arbitrary presortings -- this is because $\mathcal{T}_\pi$ can be wildly different for different $\pi$, and it is a priori not clear how to efficiently construct $\mathcal{T}_\pi$ given $\pi$.
Nonetheless, the idea behind the linear-depth ADT $\mathcal{T}_\pi$ will be the first building block for our more powerful algorithm in the later \cref{sec:sub-n-log-n-algo}. We now explain in detail how to construct such an ADT. We note that the original idea originates from Buchin and Mulzer~\cite{BuchinMulzer2011Delaunay}. 
Our tree constructs the quadtree top down and recursively.
Specifically, we assume as recursive input a tuple $(P,n,B,A_x,A_y)$: 
\vspace{-0.2cm}

\begin{invariant}
    \label{inv:recursive}
    We maintain the invariant that our recursive input is a point set $P$ of size $n$, an  open, half-open or closed square $B$, and two input arrays $A_x$ and $A_y$ that contain $P$ sorted by $x$- and $y$-coordinate respectively. 
    The algorithm will perform either a Type~\ref{type:one} or Type~\ref{type:two} split on $B$, and recurse on each resulting square with two corresponding sorted arrays. 
\end{invariant}

\mysubpara{Splitting $B$.}
We describe a procedure that determines which of the Type~\ref{type:one}\ or Type~\ref{type:two}\ splits the square $B$ requires. If a Type~\ref{type:two} split is required, we create the recursive input according to Invariant~\ref{inv:recursive}. 
The main idea is to use a two-sided exponential search approach in order to make sure that this information is propagated using only sublinear depth in the ADT. Together with the recursive inequality described in \cref{sec:recursive_inequality}, this will yield the result. The procedure runs in two stages, first for splitting along $x$ and then along $y$. 

First, consider the following notation, depicted in Figure~\ref{fig:quadtree-split} (b). Let $\ell_x$  and $\ell_y$ be the vertical and horizontal lines splitting $B$ in half (geometrically).
We denote by $R^l$ and $R^r$ the two rectangles formed by partitioning $B$ along $\ell_x$.
We define $P^l =  P \cap R^l$ and $P^r = P \cap R^r$. 
Finally, we define the $2$-presortings  $(A_x^l, A_y^l, \pi^l)$ and $(A_x^r, A_y^r, \pi^r)$ of $P_l$ and $P_r$. 
Our first stage computes $R^l$, $R^r$ and these presortings the following way. 
Let $n_1 = |P^l|$ denote the number of points of $P$ that lie strictly left of $\ell_x$ (inside $R^l$).
We first consider a branching of our decision tree into three branches: where either $n_1 = 0$, $n_1 = n$, or $n_1 \in [1, n-1]$. 
Note that we can identify in which of these three cases we are  using a subtree of constant height by comparing $\ell_x$ to $A_x[1]$ and $A_x[n]$. 
If $n_1 = 0$, we conclude the first phase by reporting $R^l$ with $(A_x^l, A_y^l, \pi^l) = (\texttt{null}, \texttt{null}, \texttt{null})$ and  $R^r$ with $(A_x^r, A_y^r, \pi^r) = (A_x, A_y, \pi)$.
If $n_1 = n$, we set $(A_x^l, A_y^l, \pi^l) = (A_x, A_y, \pi)$ and  $R^r$ with $(A_x^r, A_y^r,\pi^r) = (\texttt{null}, \texttt{null}, \texttt{null})$.
In both of these cases, the branches then enter the second phase. 

In the third case, the branch where $n_1 \in [1, n-1]$, we place a subtree that simulates exponential search over $A_x$. This search iteratively compares for an index $j$ the point $A_x[j]$ to the line $\ell_x$.
This procedure finds the maximum index $i$ such that $A_x[i]$ lies left of $\ell_x$ using a path of length $O(\log (\min \{ i, n - i \}))$.
Consider the branch corresponding to some fixed $i$. In this branch, we can obtain an $x$-sorted array $A_x^l$ storing $P_l$ by splitting $A_x$ on $i$.
Obtaining $A_y^l$ is more work:  for each distinct value $i$, consider the permutation $\pi$ restricted to $[1, i]$. This gives a set of indices $I_i \subset [1, n]$ and the array $A_y^l$ equals all elements $A_y[j]$ for $j \in I_i$, where the $j$'s are ordered from high to low.
Thus, given $i$ and $\pi$, there exists some fixed injective map $\pi_i \colon [i] \rightarrow [n]$ such that $A_y^l$ is equal to the ordered sequence $A_y[\pi_i(j)]$ for $j \in [i]$. 
It follows that the branch $\mathcal{B}_i$ of the decision tree that corresponds to splitting $A_x$ on index $i$ can implicitly consider the map $\pi_i$ to obtain $A_y^l$ and thus a presorting $(A_x^l, A_y^l, \pi_i)$ of $P_l$. 
By \emph{implicitly}, we mean the following: Since $\pi$ is fixed, it means that once we are inside the branch $\mathcal{B}_i$, the permutation $\pi_i$ is unambiguously described, and contains the information about which points constitute the preordering $A^l_y[j]$ for $j \in [i]$. 
Hence in all future recursive calls, whenever the decision tree would need to reference an element $A^l_y[j]$ of the sorted order $A^l_y$, it can reference $A_y[\pi_i(j)]$ instead. Note that since we are only interested in the existence of a decision tree, and not how to efficiently construct it, this implicit operation is allowed.

Presorting $P_r$ is done analogously. It follows that in all three cases we  enter the second phase having computed $R^l$ and $R^r$ and the presortings $(A_x^l, A_y^l, \pi^l)$ and  $(A_x^r, A_y^r, \pi^r)$. The path to this branch had length $O(1 + \log (\min \{ i, n - i \})) = O(1 + \log (\min \{ n_1, n - n_1 \}))$.

\mysubpara{The second phase.}
Having halved $B$ by $\ell_x$,  our decision tree branches into the second phase. 
Let $\ell_y$ be the horizontal line that splits $B$ in half. 
The line $\ell_y$ splits $R^l$ into two squares $B_1$ and $B_2$ where we define $P_1 = P^l \cap B_1$ and $P_2 = P^l \cap B_2$ (see  Figure~\ref{fig:quadtree-split}~(b)). 
Let $n_1' = |P_1|$ denote the number of points that lie strictly below $\ell_y$ (these lie inside $P_1$).
We use the exact same technique as above to obtain a presorting $(A_x^1, A_y^1)$ of $P_1$  (and a presorting $(A_x^2, A_y^2)$ of $P_2$) after a path of depth $O(1 + \log (\min \{ n_1', |P^l| - n_1' \}))$.
We define $B_3$ and $B_4$ off of $R^r$, with $P_3 = P^r \cap B_3$ and $P_4 = P^r \cap B_4$ analogously and obtain a presorting of $P_3$ and $P_4$ using a path of depth  $O(1 + \log (\min \{ n_2', |P^r| - n_1' \}))$ where $n_2' := |P_3|$. 

Denote $(B_1, B_2, B_3, B_4) = \texttt{split}(B)$, define $P_i := P \cap B_i$.
Denote by $n_1$ the number of points in $B_1 \cup B_2$ and by $n_2 = n - n_1$. 
Denote by $n_1'$ the number of points in $B_1$ and by $n_2''$ the number of points in $B_3$.
The above procedure defines a subtree that computes for each $i$ the point set $P_i$ and a $2$-presorting of $P_i$ using a path of depth:
\begin{equation}
\label{eq:depth}
    O \left( 1 + \log  (\min \{ n_1, n_2 \}) + \log  (\min \{ n_1 - n_1', n_1' \}) + \log  (\min \{ n_2 - n_2', n_2' \}) \right)
\end{equation}

\paragraph*{Determining the split type and tree-depth analysis.}

Note that given $(B_1, B_2, B_3, B_4) = \texttt{split}(B)$, define $P_i := P \cap B_i$, we can determine the split type of the quadtree. 
Specifically, we perform a Type~\ref{type:one} split if and only if there exists a unique $i$ such that $P_i = P$. Our ADT defines a minimum-area closed bounding square $B' \subset B_i$ as a constant-complexity expression using $\{ A_x[1], A_x[n], A_y[1], A_y[n] \}$. We can then recurse on $(P, B', A_x, A_y)$.
Otherwise, we perform a split of Type~\ref{type:two}.
Per construction, our ADT has obtained a $2$-presorting of each point set $P_i = P \cap B_i$ and so this also maintains our input invariant. 
What remains is to analyse the total depth of this tree.

First, we note that splits of Type~\ref{type:one} increase the depth of our tree by at most a constant.
Specifically, Type~\ref{type:one} occurs if and only if $n_1, n_2 \in \{ n, 0 \}$ and $n_1', n_2'  \in \{ n, 0 \}$ and so identifying a Type~\ref{type:one} split takes constant depth. 
We cannot compute two Type~\ref{type:one} splits in a row as the new bounding box $B'$ has two points of $P$ on opposite facets. I.e., the horizontal or vertical middle of $B'$ must separate at least two points of $P$. 

To bound the tree depth, we switch our perspective slightly. 
Let us denote by $T_x(n)$ the remaining depth of the ADT when considering the point set $P$ of size $n$ (before performing the $x$-split). 
Let us further denote by $T_y^l(n_l)$ the remaining depth of the ADT when considering the point set $P^l$ of size $n_l$. Likewise denote by $T_y^r(n_r)$ the remaining depth of the ADT when considering the point set $P^r$ of size $n_r$. 
For some $n \in \N$, let $\hat{T_x}(n)$  denote the maximum of ${T_x}(n)$ over all point sets $P$ of size $n$ that can be encountered throughout the ADT. Likewise, let $\hat{T}_y^l(n)$ ($\hat{T}_y^r(n)$ respectively) denote the maximum of ${T}_y^l(n)$ over all point sets $P^l$ of size $n_l$ (the maximum of ${T}_y^r(n)$ over all point sets $P^r$ of size $n_r$ respectively).
We have 
\begin{align*}
\hat{T_x}(n) \leq \max_{1 \leq n_1 < n} \, \hat{T}_y^l(n_1) + \hat{T}_y^r(n -n_1) + O(1 +\log(\min(n_1, n-n_1))),
\end{align*}
we have similar recursions for $\hat{T}_y^l(n_1), \hat{T}_y^r(n)$.
We define $T(n) := \max \, \set{\hat{T}_x(n), \hat{T}_y^l(n), \hat{T}_y^r(n)}$
\begin{align*}
\text{ and obtain: } \quad T(n) \leq \max_{1 \leq i < n} {T}(i) + {T}(n - i) + O(1 +\log(\min(i, n-i))).
\end{align*}
This recursive inequality evaluates to $O(n)$. For completeness, we show this argument in Appendix~\ref{sec:recursive_inequality}.  Therefore the clairvoyant ADT $\mathcal{T}_\pi$ has linear depth, and we conclude Theorem~\ref{thm:buchin}.

\section{A randomised $O(n\sqrt{\log n})$ time algorithm for Quadtree construction}
\label{sec:sub-n-log-n-algo}

Djidjev and Lingas~\cite{DJIDJEV1995VoronoiSorted} showed that Voronoi diagrams (and thus quadtrees) are not \texttt{1-Presortable}. 
In this section, we show that quadtrees are \texttt{2-Presortable} as we present a randomised $O(n \sqrt{\log n})$-time Real RAM algorithm to construct a quadtree from a 2-presorted point set $(A_x, A_y, \pi)$.
One of our key insights is to use the efficient orthogonal range searching data structure by Belazzougui and Puglisi \cite[Theorem~4]{belazzougui2016range}.
This is a Word RAM algorithm that, after $O(n \sqrt{\log n})$-time preprocessing, can, for any rectangular range, locate the rightmost, leftmost, topmost or bottommost point in the query. We abuse the fact that the permutation $\pi$ specifies inputs in the range $[n]$ to simulate the required Word RAM instructions at no overhead. Specifically, we show in Section~\ref{sec:auxiliary} the following:

\begin{restatable}{theorem}{simulation}
    \label{thm:simulation}
    Let $\pi$ be a size-$n$ integer array that specifies a permutation from $[n]$ to $[n]$. 
    Let $I_\pi = \{  (i, \pi(i) ) : i \in [n]\}$ be the induced point set in $[1, n]^2$.
    On a Real RAM we can, with $O(n \sqrt{\log n})$ preprocessing, store $I_\pi$ in a linear-size data structure that can answer the following queries each in $O(log^{\varepsilon}(n))$ time (for any fixed $\varepsilon>0$):
    \begin{itemize}[nolistsep]
        \item \texttt{xNext(rectangle $R \subset [1, n]^2$)} reports the leftmost and rightmost point in $I_\pi \cap R$.
        \item \texttt{yNext(rectangle $R \subset [1, n]^2$)} reports the topmost and bottommost point in $I_\pi \cap R$.
    \end{itemize}
\end{restatable}

Assuming the existence of the auxiliary data structure, we compute from a $2$-presorted point set $P$ given as  $(A_x, A_y, \pi)$ its quadtree in  $O(n \sqrt{\log n})$ expected time.
We denote by $[1, n]^2 \cap \mathbb{N}^2$ the \emph{rank space}.
The point set $P$ has a corresponding point set $R_P$ in rank space where $(i, j) \in R_P$ if there exists a point $p \in P$ where $A_x[i] = p$ and $A_y[j] = p$. 
We denote by $\rho \colon \mathbb{R}^2 \rightarrow [1, n]^2 \cap \mathbb{N}^2$ the map that maps any $p \in P$ to its corresponding point in rank space.

\mysubpara{Morphing squares into rank-space rectangles.}
Let $B \subset \mathbb{R}^2$ be any axis-aligned square and let $(A_x,A_y,\pi)$ be a presorting of a planar point set $P$.
Let $p_\ell, p_r, p_b,$ and $p_t$ denote the leftmost, rightmost, bottommost, and topmost points of $B \cap P$, respectively.
We write $\texttt{xIndex}(p)$ for the index of a point $p$ in $A_x$, and $\texttt{yIndex}(p)$ for its index in $A_y$.
Observe that every point $p \in P \cap B$ satisfies
$\texttt{xIndex}(p) \in [\texttt{xIndex}(p_\ell), \texttt{xIndex}(p_r)]$
and
$\texttt{yIndex}(p) \in [\texttt{yIndex}(p_b), \texttt{yIndex}(p_t)]$.
Conversely, any point whose indices lie in these two intervals lies in $B$.
We therefore associate with $B$ the axis-aligned rectangle in rank space
$\Gamma(B) :=
[\texttt{xIndex}(p_\ell), \texttt{xIndex}(p_r)] \times
[\texttt{yIndex}(p_b), \texttt{yIndex}(p_t)]$,
which is the unique minimum-area rectangle with: 
\[
\hfill p \in P \cap B \hfill \text{if and only if} \hfill
(\texttt{xIndex}(p), \texttt{yIndex}(p)) \in \Gamma(B) \hfill
\]

\subparagraph{Preprocessing $\boldsymbol{P}$.}
Before computing the quadtree of $P$, we preprocess $P$ in the following manner.
We create subsets $P_0 \supset  P_1 \supset P_2 \ldots$ via a stochastic process where $P_0 = P$ and $P_i$ is constructed by sampling every point from $P_{i-1}$ at probability $\frac{1}{2}$. 
We store for each of the expected $O(\log n)$ point sets $P_i$, its corresponding subset $R_{P_i} \subset R_P$ in the auxiliary data structure $D(R_{P_i})$ from Theorem~\ref{thm:simulation}.
Since the sizes of $P_i$ follow, in expectation, a geometric sum this takes expected $O(n \sqrt{\log n})$ time.

\subparagraph{Our recursive invariant. }
After preprocessing, we construct the quadtree top-down; splitting a square $B$ and recursing on its children. 
The recursive invariant used in Section~\ref{sub:adt} is too expensive to maintain on a Real RAM. Rather, we maintain a lighter invariant where we assume that our recursive input is an open, half-open, or closed square $B$ and the leftmost, rightmost, bottommost, and topmost point of $P \cap B$.
Note that we can identify that $B$ is a leaf of the quadtree by testing if these four points are the same.
From here, we simulate the algorithm from Section~\ref{sub:adt} at $O(\log^\varepsilon n)$ overhead using the queries from Theorem~\ref{thm:simulation}. 

\begin{lemma}
    \label{lem:halfsplit}
    Let $B$ be a rectangle and denote by $\ell_x$ the vertical line splitting $B$ in half.
    Denote by $R^l$ and $R^r$ the two rectangles formed by partitioning $B$ along $\ell_x$ and define $P^l = P \cap R^l$ and $P^r = P \cap R^r$.
    Given $B$, the leftmost, rightmost, bottommost, and topmost point of $P \cap B$ and our auxiliary data structures, we can determine the leftmost, rightmost, bottommost, and topmost points of $P^l$ and $P^r$ in expected $O( (1 + \log  (\min \{ |P^l|, |P^r| \} ) ) \log^{\varepsilon} n )$ time.
\end{lemma}

\begin{proof}
    Consider the sequence $\sigma$ of all points in $P \cap B$, sorted by $x$-coordinate.
    Then the leftmost and rightmost points, $p_l$ and $p_r$ in $P \cap B$, are at the start and end of $\sigma$. 
    Denote by $q$ the rightmost point in $P^l$ then $q$ is at index $b = |P^l|$ in $\sigma$. 
    We wish to find $q$ from $p_l$ in $O(\log b)$ time.   
    We achieve this by performing exponential search as in a skip list:

    Consider the point sets $P_0, P_1, \ldots $.
    The probability that a point $p \in P$ is in $P_i$ equals $2^{-i}$.
    Now consider all points in $\sigma[1, b]$. 
    We define for every such point $p$ its \emph{height} as the highest index $i$ such that $p \in P_i$.
    We denote by $p_{max} \in \sigma[1, b]$ the node from $P^l$ with the highest height $i$. In expectation, it holds that $i \in \Theta(\log b)$.
    This enables the following algorithm: 
    
   \mysubpara{Overall algorithm.}
    Given the  leftmost, rightmost, bottommost, and topmost point of $P \cap B$, we compute $\Gamma(B)$ in constant time. 
    We set $G = \Gamma(B)$ and alter it by moving its left boundary by one. 
    From $i = 0$, we query the level-$i$ data structure $D(R_{P_i})$ with $\texttt{xNext}(G)$ to find the leftmost point in $G$.     If the result lies left of $\ell_x$, we increment $i$, continuing until we find $p_{max}$.     We then shrink $G$ by setting its left boundary to $\texttt{xIndex}(p_{max}) + 1$.
    We set a counter $c$ to $i$ and decrement it once.  
    We finally apply the following skip-list algorithm (see Figure~\ref{fig:skiplist}):
    \begin{itemize}[noitemsep]
        \item While $\texttt{xNext}(G)$ on $D(R_{P_c})$ finds a point $p$ left of $\ell_x$, shrink $G$ by setting its left boundary to $\texttt{xIndex}(p) + 1$. We call this a \emph{right-step}.
        \item Otherwise, decrement $c$ by one. We call this a \emph{down-step}.
       \item If we cannot perform a right-step and $i = 0$, we have found $q$.
   \end{itemize}
   
    This procedure always finds $q$ since the left boundary of $G$ never corresponds to a point right of $\ell_x$,  and $q \in P_0$. 
    We now argue that the expected running time of this procedure is only $O(\log b)$ plus $O(\log b)$ queries to the data structure. We will use very similar arguments as in the original analysis of the skip list data structure introduced by Pugh \cite{pugh1990skip}.
    First, observe that the procedure first finds $p_{max}$. Then the counter $c$ (which can be interpreted as the current height of the search) only ever decreases until $q$ is found. Note that the height $i$ of $p_{max}$ in expectation is $O(\log b)$, but depending on the randomness used could be even larger. Based on this, we divide the procedure into three phases.
    \begin{enumerate}
        \item The phase until $p_{max}$ is found
        \item The phase during which the counter $c$ satisfies $c \geq \log b$
        \item The remaining phase where $c < \log b$ until $q$ is found.
    \end{enumerate}
    We limit the expected running time of each phase separately. Phase one has running time  $O(i)$ where $i$ is the height of $p_{\max}$, hence $O(\log b)$ in expectation.
    
    The running time of phase two is equal to the number $M_1$ of times $c$ was decremented during phase two, plus the number $M_2$ of right steps performed during phase two. We have $M_1 \leq i$, hence $\E[M_1] \leq \log b$. 
    Furthermore, $M_2$ is bounded by the number of elements in the array $\sigma[1,b]$ whose height is $\log b$ or higher. Hence we have $\E[M_2] \leq b2^{-\log b} = 1$.

    Finally, the running time of phase three is analysed using a classical skip-list argument: For each $k = \log b, \log b -1, \dots, 1$, 
    let $Y_k$ be the number of steps the algorithm performs during phase three from meeting the first element at height $k$ until meeting the first element at height $k-1$. Pugh showed that $\E[Y_k] = O(1)$ \cite{pugh1990skip}. 
    Indeed, we can imagine the search occurring backwards, and the randomness that generates the height of the elements being generated on-the-fly. 
    Under this interpretation, if the current element has height $k-1$, and we consider some element $x$ before the current element with height $k'$, we have a backwards pointer if and only if $k' \geq k-1$ and we go up a level if and only if $k' \geq k$. Then for each backwards pointer, the probability that we go up a level is at least $1/2$. This implies that $\E[Y_k] = O(1)$. In total, the expected runtime of phase three is $\sum_{k=2}^{\log b} \E[Y_k] = O(\log b)$.

    We can perform a symmetric procedure from $p_r$. Specifically, we can query  $\texttt{xNext}(G)$ to find the rightmost point in $G$ and altering the right boundary instead, to find $q$ using in expectation $O(1 + \log  (\min \{ |P^l|, |P^r| \} ) )$ queries.  
    Given $q$, we have found the leftmost and rightmost points of $R^l$.
    Doing a symmetric procedure gives the leftmost and rightmost points in $R^r$.
    Finally, given $q$, we partition $\Gamma(B)$ into two rectangles where the first rectangle $R_1$ is $\Gamma(B)$ up to $\texttt{xIndex}(q)$ and the second $R_2$ is $\Gamma(B)$ from $\texttt{xIndex}(q) + 1$ rightwards. 
    We query these two rectangles with the \texttt{yNext} queries on $A(P_0)$ to find the topmost and bottommost points in $R^l$ and $R^r$.
\end{proof}

\begin{theorem}
    \label{thm:quadtree}
    Let $P$ be a size-$n$ point set given as $(A_x, A_y, \pi)$ where $A_x$ and $A_y$ sort $P$ by $x$ and $y$ respectively, and $\pi$ is the permutation that maps $A_x$ to $A_y$. 
    We can construct a quadtree on $P$ in expected $O(n \sqrt{\log n})$ time.
\end{theorem}

\begin{proof}
    Our preprocessing required expected  $O(n \sqrt{\log n})$  time.
    We then split a square $B$, and recurse on its children maintaining our recursive invariant that for the input $B$ we have the leftmost, rightmost, bottommost, and topmost point of $P \cap B$.
    The square $B$ is a leaf of the quadtree if and only if these points are all the same (or \emph{null} for an empty square). 
    Such a split is performed in the following manner.
    Let $(B_1, B_2, B_3, B_4) = \texttt{Split}(B)$ and denote by $\ell_x$ the vertical line splitting $B$ in two and by $\ell_y$ the horizontal line.
    We use Lemma~\ref{lem:halfsplit} to run the exact algorithm as in Section~\ref{sub:adt}:

    We split $B$ into $R^l$ and $R^r$ in expected time $O( (1 + \log  (\min \{ |P \cap R^l|, |P 
    \cap R^r| \} )) \log^\varepsilon n)$.
    We then split $R^l$ into $(B_1, B_2)$ and $R^r$ into $(B_3, B_4)$ in expected time  $O(\log^\varepsilon n  \cdot (1 + \log  (\min \{ |P \cap B_1|, |P 
    \cap B_2| \} ) + \log  (\min \{ |P \cap B_3|, |P 
    \cap B_4| \} )))$.   
    If there exists an index $i$ such that $B \cap P = B \cap P_i$ we can detect this case because for all other $B_j$, the topmost point is \emph{null}. We perform a Type~\ref{type:one} split using the leftmost, rightmost, bottommost, and topmost point of $P \cap B$ and we recurse on the new child $B'$. As in Section~\ref{sub:adt}, we are guaranteed that this recursive  call generates a Type~\ref{type:two} split next. 
    For Type~\ref{type:two} splits, we recurse on at most four subproblems $B_i$ of size $|P \cap B_i|$. We can isolate the $O(\log^\varepsilon n)$ term and apply the recursive inequality from Section~\ref{sub:adt}.
    By linearity of expectation, the resulting running time after preprocessing is then $O(n \log^\varepsilon n)$ in expectation which concludes the theorem. 
\end{proof}

  \begin{figure}[h]
        \centering
        \includegraphics[width=1\textwidth]{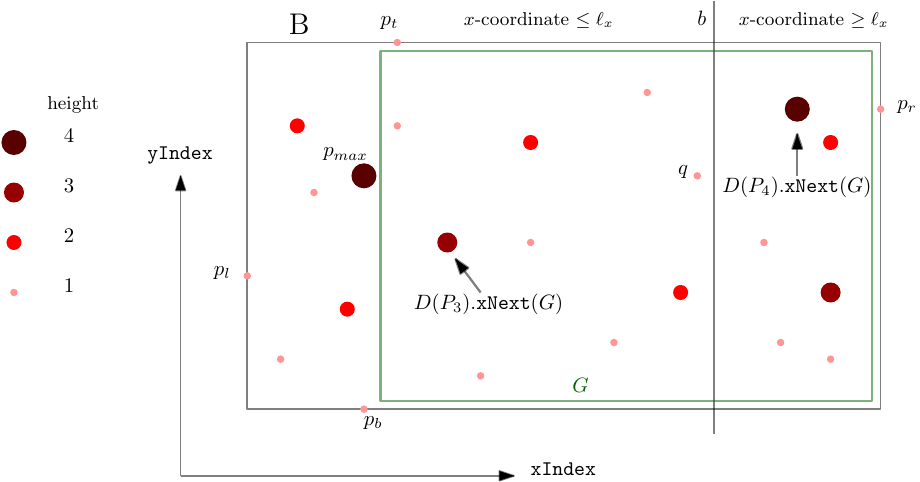}
        \caption{Visualization of our main algorithm. The auxiliary data structures $D(R_{P_i})$ for $i = 0,1,2,\dots$ are used to implicitly represent a skip list.}
        \label{fig:skiplist}
    \end{figure}

\begin{corollary}[Equivalence from~\cite{Loffler2012Triangulating}]
    \label{cor:proximity}
   Let $P$ be a size-$n$ point set given as $(A_x, A_y, \pi)$ where $A_x$ and $A_y$ sort $P$ by $x$ and $y$ respectively, and $\pi$ is the permutation that maps $A_x$ to $A_y$.      We can construct any proximity structure on $P$ in expected $O(n \sqrt{\log n})$ time.
\end{corollary}


\section{Orthogonal Range Successor for the Real RAM}
\label{sec:auxiliary}

In Section~\ref{sec:sub-n-log-n-algo} we show a Real RAM algorithm that computes a quadtree (and from it, any proximity structure) in expected $O(n \sqrt{\log n})$ time.
A central ingredient of our algorithm is a data structure that solves the \emph{orthogonal range successor problem} over some integer-valued data in the range $[n]$. 
To this end, we wish to make use of an existing Word RAM data structure by Belazzougui and Puglisi \cite{belazzougui2016range}.
However, since we are on the Real RAM, we cannot immediately deploy such data structures. 
Instead we show that, by making use of a careful analysis of the bit complexity, the Word RAM operations used by this data structure can be simulated on Real RAM at no asymptotic overhead.

\simulation*

\begin{proof}
Belazzougui and Puglisi~\cite[Theorem 4]{belazzougui2016range}
show how to do exactly this in the Word RAM model using at most $\log n+O(1)$ bits per word.
Therefore, it would suffice to simulate this Word RAM data structure
on Real RAM with constant overhead per operation,
and $O(n)$ preprocessing time.
We will do exactly this.
Note that we only need to support these operations on integers consisting
of at most $\log n+O(1)$ bits.

We use the textbook Word RAM definition by Aho, Hopcroft, and Ullman~\cite{aho1974design}.
Of these operations, the one we cannot immediately support is integer division.
We will instead start by supporting a simpler operation:
Rightward bit shifts by a fixed amount $k$ (to be chosen later),
or, equivalently, integer division by a fixed power of two $2^k$.
We denote this operation as $i>>k:=\left\lfloor\frac{i}{2^k}\right\rfloor$.
To support this operation, we will employ a variant of a standard trick
used within Word RAM algorithms
to support additional types of operations:
An operation-specific lookup table.
Let $N$ be the maximum integer value that can be represented in $\log n+O(1)$ bits.
Our lookup table will have size exactly $N$, the $i$th entry of which
contains the result of $i>>k$.
We will generate this lookup table in two phases:
\begin{itemize}
    \item In Phase 1, we compute every multiple of $2^k$, and mark the answers in the table.
    \item In Phase 2, we sweep through the table, and fill the remaining entries accordingly.
\end{itemize}
This pre-computation takes $O(N)=O(n)$ time to compute.

Now, we can use the fixed right-shift operation as follows:
For any integer $i\leq N$ in a single register in Real RAM,
we can now split it into $\left\lceil\frac{\log n+O(1)}{k}\right\rceil$ contiguous registers
each containing a sequence of chunks with $k$ bits each,
in time proportional to this number of contiguous registers.
To do this with an integer $i$,
simply compute $i-(i>>k)*2^k$ to get the lowest chunk of $k$ bits,
and then take $i\leftarrow i>>k$ to delete these extracted bits.
These registers can be easily re-combined into one integer using multiplication and addition.

Now, to support integer division, we can employ a more standard variant of the same trick.
First, we will brute-force a look-up table for small values, including remainder results.
Second, we will essentially combine these small computations on chunked integers
using basic long division.
Pick $k=\frac{\log N}3$.
We can brute-force compute a table that allows us to do integer division
on integers of size $\leq k$.
We can afford to do each of these computations bit-by-bit
(so $O(k^2)$ time per operation),
since there is a small number of combinations ($O(2^{2k})$ in this case).
We can do the same to compute modulo on integers of size $\leq k$.
We can then combine the results using standard big-integer division
methods (e.g., Knuth's Algorithm D~\cite{knuth2014art}).
\end{proof}

\noindent
Given integer division, all other typical Word RAM operations with a constant number of operands
(such as bitwise \texttt{AND}, \texttt{OR}, \texttt{XOR}, and unary operations like \texttt{MSB})
can be simulated~\cite{Fredman1993Surpassing}.
The construction is quite similar to the above proof.
This is the standard approach used for extending the operation set in Word RAM.

\section{Expanding the Population of our Hierarchy}

In this section, we populate our hierarchy further by classifying several geometric problems.
First, we show that computing a triangulation is \texttt{1-Presortable}.
Next, we show that computing a KD-tree, the orthogonal segment intersection detection problem, and the maximum-area empty circle problem are \texttt{2-Presortable}. We do this by showing that they are not \texttt{1-Presortable} and providing randomised $o(n \log n)$-time algorithms that takes presortings as input.
In Section~\ref{sec:lowerbounds}, we show that several problems are \texttt{Presort-Hard}.

\subsection{Triangulations are \texttt{1-Presortable}}

We consider the problem of finding \emph{some} triangulation of a planar point set.
Although both the Triangulation and Delaunay Triangulation problems
have tight $\Theta(n\log n)$ algorithms in general,
it turns out that Triangulation is easier than Delaunay Triangulation
in the context of our presorting hierarchy. For triangulations we assume the input to lie in general position to avoid degeneracies:

\begin{theorem}
\label{thm:a-triangulation}
Let $P$ be a size-$n$ point set in general position where no two points share an $x$-coordinate.
Given an array $A_x$ that stores $P$ in sorted order, we can compute a triangulation of $P$ in linear time. 
\end{theorem}

\begin{proof}
Since the input is sorted by $x$-coordinate, 
we can compute the edges of the convex hull of $P$
in linear time~\cite{andrew1979another}.
Since no two points have the same $x$-coordinate,
$A_x$ also forms an $x$-monotone polygonal chain.
If we take the union of the edges in each of these structures,
we obtain a connected plane graph.
For each face of this plane graph,
we can then triangulate it with Chazelle's linear-time
triangulation algorithm for simple polygons~\cite{chazelle1991triangulating}.
See \cref{fig:a-triangulation} for an example of this construction.
\end{proof}

\begin{figure}
    \includegraphics[page=2,width=0.3\linewidth]{graphics/a-triangulation}
    \hfill
    \includegraphics[page=3,width=0.3\linewidth]{graphics/a-triangulation}
    \hfill
    \includegraphics[page=1,width=0.3\linewidth]{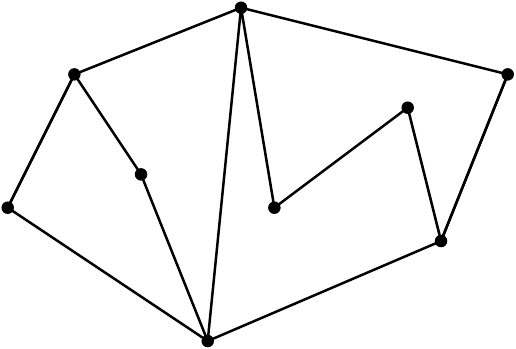}
    \caption{The construction for computing some triangulation, given $A_x$.
    We take the union of the convex hull (left) and an
    $x$-monotone polygonal chain through all the vertices (middle)
    to arrive at a connected plane graph whose vertices are exactly $P$ (right).}
    \label{fig:a-triangulation}
\end{figure}

\subsection{Constructing a KD-tree}

A KD-tree is a planar data structure that can be defined as follows: 
The root node represents a point set $P$. 
It splits $P$ by $x$-coordinate along the point with median $x$-coordinate.
Its child nodes split their remaining points by median $y$-coordinate instead and this recurses until the leaves store a single point.
We show that a KD-tree is in the class of \texttt{2-Presortable} problems, employing a similar technique to that in Section~\ref{sec:sub-n-log-n-algo}:


\begin{theorem}
\label{thm:kd-tree}
Let $P$ be a size-$n$ point set. Consider $(A_x, A_y, \pi)$  where $A_x$ and $A_y$ sort $P$ by $x$ and $y$ respectively, and $\pi$ is the permutation that maps $A_x$ to $A_y$.
Given only $A_x$, there is a comparison-based $\Omega(n \log n)$ lower bound for constructing the KD-tree of $P$.
Given $(A_x, A_y, \pi)$, we can construct a KD-tree in expected $O(n \sqrt{\log n})$ time. 
\end{theorem}

\begin{proof}
First we show the lower bound: for a set of $1$-dimensional points along the $y$-axis,
a KD-tree is exactly a (balanced) binary search tree
over the points.
It follows that we can sort any set of points $P$ by their $y$-value using a linear-time in-order traversal of the KD-tree of $P$. This gives an immediate reduction from sorting: given any set of values $V = (v_1, \ldots, v_n)$ create the point set $P = \{ (i, v_i) \mid v_i \in V \}$.
Then $A_x = V$ is an array storing $V$ sorted by $x$-coordinate. Yet, any $o(n \log n)$-time algorithm to construct a KD-tree implies an $o(n \log n)$-time algorithm for sorting $V$ by the in-order traversal of the output. 

Given a 2-presorting $(A_x, A_y, \pi)$ we obtain an algorithm with a running time of expected $O(n \sqrt{\log n})$ by running the algorithm that is used in  Theorem~\ref{thm:quadtree}.
Recall that we denote by $[1, n]^2 \cap \mathbb{N}^2$ the \emph{rank space}.
The point set $P$ has a corresponding point set $R_P$ in rank space where $(i, j) \in R_P$ if there exists a point $p \in P$ where $A_x[i] = p$ and $A_y[j] = p$. 

In the proofs supporting Theorem~\ref{thm:quadtree},
we first preprocess $R_P$ by storing it in the orthogonal range successor data structure by Belazzougui and Puglisi~\cite{belazzougui2016range} (abusing that coordinates in $R_P$ are integers in $[n]$ to simulate this Word RAM algorithm on a Real RAM). 
Our algorithm then recursively halves a rectangle $R \subset \mathbb{R}^2$ by either a vertical line $\ell_x$ or horizontal line $\ell_y$. 
Say we split on $\ell_x$. 
We partition the points in $P \cap R$ along this line implicitly, using an implicit skip-list to find the rightmost point left of $\ell_x$ using exponential search.
Our implicit skip-list algorithm executes the skip-list algorithm at $O(\log^\varepsilon n)$ overhead by querying the data structure to get the successor in the skip list.
We then compare the successor to $\ell_x$ in constant time to test whether we need to go right or down in the skip list. 

To construct a KD-tree, we can use a very similar procedure.
Chan and Pătraşcu~\cite{Chan2010Counting} show, on a Word RAM, how to preprocess a point set in $O(n\sqrt{\log n})$ time to answer orthogonal range counting queries in $O(\sqrt{\log n})$ time.
In particular, their algorithm uses only $\log n$ bits (see the remark at the end of Section 1.3 in their paper).
Hence, via our exact same technique, we can store $R_P$ in this data structure also: simulating the Word RAM instructions on a Real RAM since $R_P$ contains only coordinates in $[n]$.
Whenever we want to split a rectangle $R$, we want to do so by the median $x$-coordinate (or $y$-coordinate).
Via an orthogonal range counting query on $\Gamma(R)$, we know how many points there are in $P \cap R$. 
We now use almost the exact same exponential search as before, encountering points $p \in R$.
However, instead of comparing $p$ to a line we instead consider splitting $\Gamma(R)$ on the $x$-value of $p$. 
This corresponds to a rectangle $R' \subset R$ in the plane, and to a rectangle $G = \Gamma(R')$ in rank space. 
By performing an orthogonal range counting query on $G$, we count how many points there are in $R' \cap P$. 
This way, we can test whether $p$ precedes, succeeds, or is the point with median $x$-coordinate and thus we can execute our exponential search.
Invoking this query takes $O(\sqrt{\log n})$ time per step in the skip-list (as opposed to the previous $O(\log^\varepsilon n)$ time). This does not asymptotically increase the overall running time. 
\end{proof}

\subsection{Orthogonal Segment Intersection Detection}

For a set of horizontal and vertical line segments,
we consider the problem of detecting an intersection between any pair of them.
In this problem, we do not have a set of points,
so it may not be immediately obvious what a presorting
is.
In this case, we consider a presorting of all the endpoints.
This intuitively matches up with the standard $\Theta(n\log n)$-time
algorithms for this problem and its generalizations~\cite{shamos1976geometric},
which all start by sorting the endpoints along the $x$-axis.
We will take a related approach to prove the following result:

\begin{theorem}
\label{thm:isect}
Let $S$ be a set of line segments and $P$ be a size-$n$ point set corresponding to all endpoints of the line segments.
Consider $(A_x, A_y, \pi)$  where $A_x$ and $A_y$ sort $P$ by $x$ and $y$ respectively, and $\pi$ is the permutation that maps $A_x$ to $A_y$.
Given only $A_x$, there is a comparison-based $\Omega(n \log n)$ lower bound for detecting if any two segments in $S$ intersect.
Given $(A_x, A_y, \pi)$, we can detect if $S$ contains any intersections in $O(n\log\log n)$ time.
\end{theorem}

\begin{proof}
For the lower bound, we reduce from the element distinctness problem:
We are given a multiset $V$ of $n$ real values and need to output whether $V$ contains any duplicate elements.
This problem is known to have an $\Omega(n\log n)$ comparison-based lower bound by Fredman~\cite{fredman_how_1976}.

We can construct a set $S$ of line segments as follows:
For each element $v \in V$, create a horizontal line segment with endpoints
$(0+\varepsilon_v,v)$ and $(1+\varepsilon_v,v)$,
where $\varepsilon_v<1$ is some perturbation determined by $v$.
The result is a set of horizontal line segments, where the corresponding point set $P$ even all has unique $x$-coordinates.
Two values $v_1,v_2 \in V$ are identical
if and only if their corresponding pair of horizontal line segments intersect.

For the upper bound, we consider as input a 2-presorting $(A_x, A_y, \pi)$ of the endpoints $P$ of segments in $V$.
We assume that each value in $A_x$ (and $A_y$) specifies a segment endpoint and which segment of $S$ it belongs to. 
Observe that we can detect intersections
between pairs of horizontal segments and pairs of vertical segments
in $O(n)$ time:
Simply bucket the horizontal segments by their $y$-coordinate
in $x$-sorted order per-bucket,
and then sweep through each bucket
to detect intersections.
Do the symmetric routine for the vertical segments.

What remains is to detect intersections
between a pair of segments where one is horizontal and the other is vertical. 
We perform a line sweep from left-to-right,
with events corresponding to each endpoint.
At all times, we will maintain a data structure
consisting of all the current horizontal segments,
in order.
We store the $y$-coordinates of these horizontal segments.
In order to detect if a vertical segment
intersects any of the current horizontal segments,
it suffices to be able to perform predecessor and successor search
on integers of size at most $O(n)$.
This is exactly the task accomplished by
a van Emde Boas tree~\cite{van1975preserving} which uses $O(\log\log n)$ time per query
and update.
\end{proof}

\subsection{Maximum Empty Circle}

The final problem that we consider is the maximum empty circle problem.
Given a point set $P$, we wish to compute a maximum-area circle which has its centre in the convex hull of $P$ but contains no point of $P$~\cite{Toussaint1983}.
It is a classic geometric problem with an $O(n \log n)$ time upper bound and a matching comparison-based lower bound.

\begin{theorem}
\label{thm:maximum}
Let $P$ be a size-$n$ point set. Consider $(A_x, A_y, \pi)$  where $A_x$ and $A_y$ sort $P$ by $x$ and $y$ respectively, and $\pi$ is the permutation that maps $A_x$ to $A_y$.
Given only $A_x$, there is a comparison-based $\Omega(n \log n)$ lower bound for computing a maximum empty circle of $P$. 
Given $(A_x, A_y, \pi)$, we can compute a maximum empty circle in expected $O(n \sqrt{\log n})$ time. 
\end{theorem}

\begin{proof}
    For the lower bound, we reduce from the \emph{largest gap} problem.
    In this problem, the input is a set of $n$ values $V$.
    Consider the set $S = (s_1, \ldots, s_n)$ which sorts $V$.
    The maximum gap is the index $i$ that maximizes $s_{i+1} - s_i$.
    This problem was shown to have a comparison-based $\Omega(n \log n)$ lower bound by Lee and Wu~\cite[Lemma A.2]{Lee1986Gap}. Moreover, as the proof is combinatorial it still holds even if we assume some minimal separation between the values in $V$. So let $V$ be a set of values where all values differ by at least 1. Denote by $S$ its unknown sorted set. 
    Given $V$, we construct in linear time a point set $P$ by constructing for each point $v_i \in V$ two points: $(-\varepsilon_v, v)$ and $(+\varepsilon_v, v)$ where $\varepsilon_v < 0.001$ is some arbitrary value.
    The result is a column of 4-gons whose width is less than $0.001$ and whose height corresponds to the gaps in $S$.
    Since the $x$-coordinates are arbitrary, an array $A_x$ that sorts $P$ by $x$-coordinates gives zero additional information. 
    The convex hull of these points is a rectangle of width $0.001$ and the largest disk that has its centre in this rectangle (but contains no point of $P$) has a diameter which corresponds to the largest gap in $S$. 
    Given a maximum empty circle, we can find its closest four points in $P$ in  linear time and thus the values in $V$ that realise the maximum gap in $S$. 

    For the upper bound, Toussaint~\cite{Toussaint1983} showed an algorithm that given the Voronoi diagram of a point set computes the maximum empty circle in linear time.
    By Corollary~\ref{cor:proximity}, given $(A_x, A_y, \pi)$, we can compute the convex hull of $P$ in expected $O(n \sqrt{\log n})$ time, which concludes the proof.   
\end{proof}

\section{Showing Presort-Hardness}
\label{sec:lowerbounds}

In this section, we show that computing an onion layer decomposition is \texttt{Presort-Hard}.
We furthermore show that constructing the ordered $k$-closest pair is \texttt{Presort-Hard}.
These results have several consequences.
Firstly, our lower bound for onion layer decompositions implies a lower bound for decremental convex hulls.
Our lower bound for the $k$-closest pair problem implies a lower bound for decremental closest pair. 

\subsection{Onion Layers}

\begin{figure}[h]
\begin{subfigure}{0.45\linewidth}
    \includegraphics[width=1\linewidth,page=1]{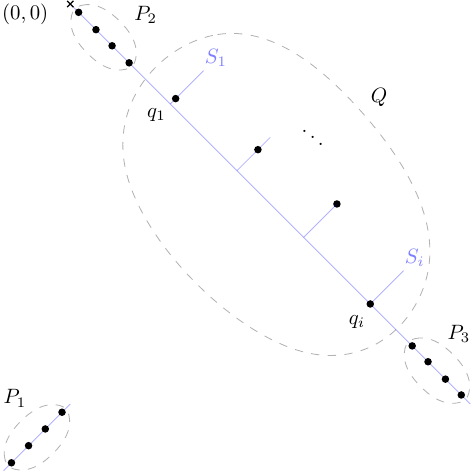}
    \caption{Point placement}
    \label{subfig:onion-layers-1}
\end{subfigure}
    \hfill
\begin{subfigure}{0.45\linewidth}
    \includegraphics[width=1\linewidth,page=2]{graphics/onion-layers-lower-bound.pdf}
    \hfill
    \caption{Onion layers}
    \label{subfig:onion-layers-2}
\end{subfigure}
    \caption{Example of a family of point sets that admits $\Omega((n/4)!)$ different onion layer decompositions, but has a fixed $x$- and $y$-sorted order.}
    \label{fig:onion-layers}
\end{figure}

In this section, we show that computing the onion layer decomposition is \texttt{Presort-Hard}. 
Assume we have a set of $n$ points $P \subseteq \R^2$, given to us as input in form of a list $L_P = (p_1,\dots,p_n)$. 
Let us call the \emph{combinatorial onion layer decomposition} of $L_P$ a list of sets $(S_1,\dots,S_k)$, such that the sets are disjoint, 
their union equals $P$, and each set $S_i$ contains exactly the vertices of the convex hull of $P \setminus (S_1 \cup \dots \cup S_{i-1})$ for $i \in [k]$. 
It is clear that each list $L_P$ has a unique combinatorial onion layer decomposition.

\begin{theorem}
\label{thm:onion_layer}
Computing the combinatorial onion layer decomposition is \texttt{Presort-Hard}.
\end{theorem}

\begin{proof}
    Let $n \in \N$ be divisible by four w.l.o.g. We will showcase a family $\mathcal{F}_n \subseteq (\R^2)^{n}$, such that each member $L_P \in \mathcal{F}_n$ is a list of $n$ points in $\R^2$, such that across $\mathcal{F}_n$ at least $\Omega((n/4)!)$ different combinatorial onion layers occur. 
    Furthermore, we ensure that each member of $\mathcal{F}_n$ has the same sorted $x$- and $y$- order. 
    In other words, there exist permutations $\pi_x,\pi_y$ of $[n]$ such that  for each list $L_P \in \mathcal{F}_n$, sorting the list $L_P$ by $x$-coordinate yields $\pi_x$, and sorting it by $y$-coordinate yields $\pi_y$.
    We claim that this suffices to prove the theorem. Indeed, since every $L_P \in \mathcal{F}_n$ has a fixed $x$-order and $y$-order, the integer part of the input is a constant over all $L_P \in \mathcal{F}_n$. Therefore there exists a clairvoyant decision tree for the onion layer decomposition of a given list $L_P \in \mathcal{F}_n$.
    Since the tree has $\Omega((n/4)!)$ leaves, it has depth $\Omega(\log((n/4)!)) = \Omega(n\log n)$. Hence any Real RAM algorithm for the onion layer decomposition requires worst-case runtime $\Omega(n\log n)$.

    The family $\mathcal{F}_n$ is depicted in \cref{subfig:onion-layers-1} and defined as follows: There are four lists $P_1,P_2,P_3,Q$ of $n/4$ points each. 
    \begin{itemize}
        \item The list $P_1$ contains $n/4$ points in order on the open segment between $(-2,-2n-6)$ and $(0,-2n-4)$.
        \item The list $P_2$ contains $n/4$ points in order on the open segment between $(0,0)$ and $(2,-2)$.
        \item The list $P_3$ contains $n/4$ points in order on the open segment between $(2n+2, -2n-2)$ and $(2n+4,-2n-4)$.
        \item Consider the segments $S_i$ between $(2i+1,-2i-1)$ and $(2i+2,-2i)$ for all $i \in [n/4]$. Choose some point $q_i$ from each segment $S_i$ for $i \in [n/4]$, and let $Q = (q_1,\dots,q_{n/4})$.
    \end{itemize}

    We obtain different members of $\mathcal{F}_n$ by choosing different placements of the points $q_i$ on $S_i$. 
    It is clear that each member $L_P \in \mathcal{F}_n$ has the same $x$-order, namely $P_1,P_2,Q,P_3$. 
    By definition the order inside of $P_1,P_2,P_3$ is fixed.
    Note that the order within $Q$ is fixed, and this holds no matter where the points $q_i$ are placed on segments $S_i$. Analogously, the $y$-order of $L_P$ is fixed.

    Finally, we claim that $\mathcal{F}_n$ admits $\Omega((n/4)!)$ onion layer decompositions. For this, let $\pi : [n/4] \to [n/4]$ be an arbitrary permutation. Consider \cref{subfig:onion-layers-2}. We assume that $P_1,P_2,P_3$ are already placed and we have to decide on the placement of the points in $Q$.
    We consider $j_1 := \pi^{-1}(n/4)$. We place the point $q_{j_1}$ on the very end of segment $S_{j_1}$, i.e. at $(2j_1+2,-2j_1)$.
    Next, we consider the convex hull $A_1$, which has four vertices, including $q_{j_1}$. We remove these four vertices and continue recursively.
    Next, we consider $j_2 := \pi^{-1}(n/4 - 1)$. Note that the segment $S_{j_2}$ non-trivially intersects $A_1$. 
    Therefore, we can place some point $q_{j_2}$ in the relative interior of $A_1 \cap S_{j_2}$. The convex hull of the remaining points has again four vertices, including $q_{j_2}$.
    We can continue this pattern, making sure to place point $q_{j_i}$ in the relative interior of
    $A_{i-1} \cap S_{j_i}$. This way, we obtain a point list $(P_1,P_2,P_3,Q)$, with the property that the onion layer decomposition of the corresponding point set is exactly the decomposition displayed in \cref{subfig:onion-layers-2}.
    In particular, the order of $Q$ with respect to the onion layers is exactly given by $\pi$. Since there are $(n/4)!$ choices for $\pi$, there are $\Omega((n/4)!)$ different combinatorial onion layer decompositions.
\end{proof}

\subsection{Decremental Convex Hull}

    We obtain the following corollary concerning the decremental 2D convex hull problem. Here, we are given an initial set $P = \set{p_1,\dots,p_n}$ of $n$ points in 2D. 
    We are given in an online fashion a sequence $(i_1,\dots,i_k)$ of indices, specifying that the point $p_{i_j}$ should be deleted from the current set for $j=1,\dots,k$. Formally, we let $P_0 := P$ and $P_j = P_{j-1} \setminus \set{p_{i_j}}$.
    At each point in time, we require the algorithm to dynamically maintain the 2D convex hull of $P_j$.
    When we say that this problem is \texttt{Presort-hard}, we mean that for $k = \Theta(n)$, no $o(n\log n)$ algorithm can correctly maintain the convex hull of the point set for all $j=1,\dots,k$, even if the presorted $x$- and $y$-order of $P$ is given as part of the input.

    \begin{corollary}
    \label{cor:decremental-convex-hull}
    The decremental 2D convex hull problem is \texttt{Presort-Hard}.
    \end{corollary}
    \begin{proof}
        If the problem had an $o(n \log n)$ algorithm given an $x$- and $y$-presorted point set, we could solve the combinatorial onion layer decomposition in $o(n \log n)$ time. We can compute an onion layer decomposition of $P$, by repeatedly querying the convex hull of $P$, and removing all of its vertices.
    \end{proof}

\subsection{Ordered $k$-closest Pair and Decremental Closest Pair}
In the ordered $k$-closest pair problem, we are given $n$ points $p_1,\dots,p_n$ and some parameter $k$. The problem is to determine from the set $\set{\left\lVert p_i - p_j \right\lVert : i,j \in [n], i\neq j}$ of all pairwise distances the $k$ smallest ones.

\begin{theorem}
\label{thm:k-closest-pair}
    The ordered $k$-closest pair problem is \texttt{Presort-Hard}.
\end{theorem}
\begin{proof}
    We show this using the sorting lower bound. Assume we can solve the ordered $k$-closest pair problem in time $o(n \log n)$ given the presorted point set. Assume further we are given a list of $n$ real numbers $x_1,\dots,x_n \in [0,1]$ which we wish to sort. We consider the line segment $\ell$ from $(0,0)$ to $(n, n)$. 
    We mark $n$ equidistant locations on $\ell$ (see \cref{fig:closest-pair}). In the vicinity of the $i$-th location we place two points $p_i$ and $p_i'$ for $i=1,\dots,n$, so that both $p_i,p_i' \in \ell$ and $\left\lVert p_i - p_i' \right\lVert = \varepsilon x_i$ for some constant small $\varepsilon > 0$ (and $p_i$ is closer to the origin than $p'_i$). Note that both the $x$- order and $y$
    order is constant, independent of the values of $x_1,\dots,x_n \in [0,1]$. So by our assumption we can recover the $n$ closest pairs in $o(n \log n)$ time. But then we have sorted $x_1,\dots,x_n$ in $o(n \log n)$ time, a contradiction.
    \begin{figure}
        \centering
        \includegraphics[width=0.25\linewidth,page=2]{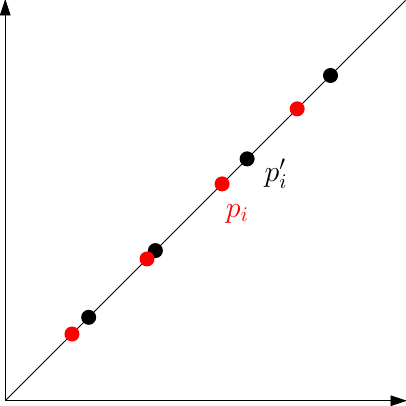}
        \caption{Lower bound construction for ordered $k$-closest pair}
        \label{fig:closest-pair}
    \end{figure}
\end{proof}

\begin{corollary}
\label{cor:decremental-closest-pair}
    The decremental closest pair problem is \texttt{Presort-Hard}.
\end{corollary}
\begin{proof}
    By the same construction as in \cref{thm:k-closest-pair}. Once we have obtained the closest pair, we delete both points of the pair. This way we sort $x_1,\dots,x_n$. Hence the decremental closest pair problem cannot be solved in $o(n \log n)$ time, even if the point set is preordered in both $x$- and $y$-direction.
\end{proof}

\section{Conclusion and Open Problems}

We conclude by mentioning some possible avenues for future work.

The most approachable open question appears to be derandomising our main result for constructing Quadtrees.
We note that this is the only barrier to derandomising the proximity structures: NN-graphs, Delaunay Triangulations, Euclidean MSTs, Voronoi diagrams, and finding the Maximum Empty Circle,
since the reductions for these problems are deterministic~\cite{Loffler2012Triangulating}.

Another interesting open question is whether the time complexity for constructing Quadtrees
(and hence all the other problems as well) can be improved to $o(n\sqrt{\log n})$.
The main barrier of our construction is the use of orthogonal range predecessor search,
which currently requires $O(n\sqrt{\log n})$ pre-computation time, even on Word RAM.
An improved result for orthogonal range predecessor search on Word RAM with the same word-size properties would hence imply an improved result on Real RAM for all of these problems.

We also ask whether we could obtain a similar result for Quadtrees or other problems
in higher dimensions.
We note that one possible avenue for this would be to construct a data structure
for the Range Minimum Query problem over points of magnitude at most $O(n)$ in $d$-dimensions (for $d>2$),
with preprocessing time $o(n\log n)$ and query time $o(\log n)$.

\bibliographystyle{plainurl}
\bibliography{refs}

\appendix

\section{Recursive Inequality}
\label{sec:recursive_inequality}
The central algorithm in our paper is analyzed using the following recursive inequality. 
Consider a recursive procedure, that operates over an array of size $N$. For convenience of notation, let in this section denote $[a,b] := \set{a, a+1,\dots, b}$. When called on a subarray $[a,b]$ of size $n := b - a + 1$, the procedure selects some splitting index $a \leq i < b$, and is afterwards called on the subarrays $[a, i]$ and $[i+1, b]$. 
The splitting index $i$ depends on $a,b$ in some arbitrary fashion.
We are interested in the runtime $T(N)$ of the procedure. Let $n := b-a+1$ and $n_1 := i - a + 1$ and $n_2 = b - i$ be the sizes of the involved subarrays. 
We assume that during each recursive call the additional work done by the procedure depends only on the \emph{minimum} of $n_1, n_2$ times a small multiplicative overhead factor $M$. That is, we have $T(1) = O(1)$ and  for some function $f : \N \to \N$ and some $M > 0$ for all $n > 1$
\begin{equation}
    T(n) \leq T(n_1) + T(n_2) + M \cdot f(\min(n_1,n_2)).
    \label{eq:recursive-inequality}
\end{equation}

In our main theorem, we will let $M = \log^\varepsilon N$.
A technical drawback of the above equation is that we can only consider such functions where each recursive call handles one contiguous subarray. 
For our main theorem, we require a  slightly more general setting. Namely, assume that $T : \N \to \R_{\geq 0}$ is a function that satisfies $T(1) = O(1)$ and for all $n > 1$
\begin{equation}
    T(n) \leq \max_{1 \leq i \leq n-1} T(i) + T(n-i) + Mf(\min(i, n-i)).
    \label{eq:recursive-inequality-2}
\end{equation}
The following lemma is a known fact \cite{BuchinMulzer2011Delaunay}, but we repeat its proof here for convenience of the reader. 
The main insight is that assuming $f$ is strictly sublinear, the runtime $T(n)$ is surprisingly low, no matter how unbalanced the splitting indices are chosen.

\begin{lemma}
\label{lem:recursive-inequality-deterministic}
    If $T(n)$ satisfies \eqref{eq:recursive-inequality} or \eqref{eq:recursive-inequality-2} and $f(x) \leq Cx^\alpha$ for some constants $0 <C, 0 < \alpha < 1$, then $T(n) \leq M \cdot O(n)$.
\end{lemma}

\begin{proof}
    First, we argue that we can w.l.o.g. only consider some function $T(n)$ defined over contiguous subarrays such that \eqref{eq:recursive-inequality} is satisfied. We claim that this implies the more general case of \eqref{eq:recursive-inequality-2} as well. 
    Indeed, assume we have some function $T(n) : \N \to \R_{\geq 0}$ that satisfies \eqref{eq:recursive-inequality-2}. For each $n$, there exists some $i_n$ that maximizes the right-hand-side of \eqref{eq:recursive-inequality-2}.  
    We imagine the recursive procedure $T'$ operating over the array $[1,N]$, such that whenever this procedure is called on some subarray $[a,b]$ of size $n = b-a+1$, it decides which splitting index to select only based on $n$.
    Namely, it selects splitting index $a + i_n - 1$. 
    One can then prove by induction that $T'$ is an upper bound for $T$ and satisfies \eqref{eq:recursive-inequality}. Hence it suffices to prove our theorem only for recursive procedures defined on contiguous subarrays satisfying \eqref{eq:recursive-inequality}.

    Let now $T(n)$ be some function that satisfies \eqref{eq:recursive-inequality}.
    We imagine the procedure running on the complete array $[1,N]$. Every time the procedure is called on some subarray $[a,b]$, we have to perform work $Mf(\min(n_1, n_2))$ where $n_1,n_2$ are determined by the splitting index $i$ corresponding to that subarray $[a,b]$.
    We consider a charging argument. If $n_1 < n_2$, we charge the total work $Mf(n_1)$ to the left half $[a,i]$ of the array. 
    More precisely, we charge each of the $n_1$ individual elements of $[a, i]$ with a charge of $Mf(n_1)/n_1$.
    In the other case $n_2 \geq n_1$, we distribute the total work $Mf(n_2)$ by charging each 
    of the $n_2$ individual elements of $[i+1,b]$ with a charge of $Mf(n_2)/n_2$. Let $\mathcal{S}$ be the set of all subarrays that are charged when computing the function for the complete array $[1,N]$, and let for each subarray $S \in \mathcal{S}$ denote $\ell_S$ the number of elements of that subarray.
    Then, by the above explanation
    \[
    T(N) = \sum_{S \in \mathcal{S}}Mf(\ell_S) = \sum_{S \in \mathcal{S}}\sum_{i=1}^{\ell_S} Mf(\ell_S)/\ell_S.
    \]
    We now count this quantity in a different way. Consider some element $j \in [1, N]$ that was charged at least once and let $\mathcal{S}(j) := \set{S \in \mathcal{S} : j \in S}$ be the set of corresponding subintervals that were charged whenever $j$ was charged.
    W.l.o.g. let $\mathcal{S}(j) = \set{S_0,\dots, S_k}$ for some $k$, where $\ell_{S_k} \geq \ell_{S_{k-1}} \geq \dots \geq \ell_{S_0}$.
    Note that every time $j$ gets charged, it must be on the smaller side of the corresponding subinterval. This means that if $j$ is charged twice in succession, its subinterval must shrink by a factor of 2 or more.
    In other words, $\ell_{S_i} \geq 2\ell_{S_{i-1}}$ for all $i = 1,\dots,k$. Since $\ell_{S_0} \geq 1$ we have $\ell_{S_i} \geq 2^i$. Then we can compute the total charge by summing the contribution of each element $j$.
    \begin{align*}
        T(N) & = \sum_{j=1}^N\sum_{S \in \mathcal{S}(j)}Mf(\ell_S)/\ell_S & \leq \sum_{j=1}^N\sum_{S \in \mathcal{S}(j)} MC\ell_S^{\alpha -1 } \leq \sum_{j=1}^N\sum_{i=0}^\infty MC2^{i(\alpha - 1)}\\
        && = \sum_{j=1}^NMC \frac{1}{1 - 2^{\alpha - 1}} = M \cdot O(N).
    \end{align*}
    This was to show.
    \end{proof}

\begin{corollary}
 If $T(n)$ satisfies \eqref{eq:recursive-inequality} or \eqref{eq:recursive-inequality-2}, and $f$ is the runtime of a randomized procedure with $\E[f(x)] \leq Cx^\alpha$, then $\E[T(n)] \leq M\cdot O(n)$.
\end{corollary}
\begin{proof}
    This follows from \cref{lem:recursive-inequality-deterministic} and the linearity of expectation. More precisely, even if the running times $T(n)$ and $f(x)$ are random variables, it is still true that $T(n) = \sum_{S \in \mathcal{S}} Mf(\ell_S)$. Then we have
     \begin{align*}
        \E[T(N)] & = \sum_{S \in \mathcal{S}}M\E[f(\ell_S)] \leq \sum_{S \in \mathcal{S}} MC\ell_S^{\alpha} \leq M\cdot O(N),
    \end{align*}
    where the last inequality follows from the same charging argument, and is true, no matter which splitting indices are chosen.
\end{proof}

\section{Further related work}
\label{app:closely}

We briefly mention additional lines of related work.
Recently, Eppstein, Goodrich, Illickan, and To~\cite{EppsteinEtAlCCCG2025Entropy} presented algorithms for computing the convex hull whose running time depends on the degree of presortedness of the input.
Van der Hoog, Rotenberg, and Rutschmann~\cite{vanDerHoogRustenmann2025TightUniversal} showed that these bounds are tight by establishing matching lower bounds based on an entropy measure of the partial order.
Presorting can also be viewed as a form of \emph{advice} about the input.
One may distinguish between \emph{predictions}, which are approximate hints about the output, and \emph{advice}, which is guaranteed to be correct.
There is a long tradition on geometric algorithms with predictions~\cite{Melhorn2996Predictions} and we highlight the recent work of Cabello, Chan, and Giannopoulos~\cite{cabello2026Predictions}, who construct a Delaunay triangulation given a predicted triangulation that is close to the true one in terms of flip distance.
Regarding \emph{advice}, there are relatively few geometric examples. 
One notable exception is the framework of \emph{imprecise geometry}, introduced by Held and Mitchell~\cite{held2008triangulating}, in which each input point $p_i$ is known only to lie within a region $R_i$.
Such regions may be viewed as advice that constrains the space of feasible inputs.
Within this framework, efficient algorithms are known for many of the proximity structures discussed above, including convex hulls, Delaunay triangulations, and related constructions~\cite{deberg2025ImpreciseConvex,devillers2011delaunay,evans2011possible,ezra2013convex,held2008triangulating,loffler2010delaunay,loffler2013unions,loffler2025preprocessing,van2010preprocessing,van2019preprocessing,van2022preprocessing}.
Finally, since many geometric lower bounds are obtained via reductions from sorting, it is worth discussing sorting itself. 
Sorting with predictions—such as long approximately sorted subsequences—has been investigated in several models~\cite{Bai2023Predictions,Castro1992AdaptiveSorting,auger2019}, while sorting with advice in the form of partial orders or guaranteed comparisons has also received considerable attention~\cite{cardinal_sorting_2013,Haeupler25,van2024tight,van2025simpler}.

\end{document}